\documentclass[11pt]{article}

\usepackage{amsthm}
\usepackage{graphicx} 
\usepackage{array} 

\usepackage{amsmath, amssymb, amsfonts, verbatim}
\usepackage{hyphenat,epsfig,subcaption,multirow}
\usepackage[font=small,labelfont=bf]{caption}
\usepackage{nicefrac}
\usepackage{paralist}

\usepackage[usenames,dvipsnames]{xcolor}
\usepackage[ruled]{algorithm2e}

\DeclareFontFamily{U}{mathx}{\hyphenchar\font45}
\DeclareFontShape{U}{mathx}{m}{n}{
      <5> <6> <7> <8> <9> <10>
      <10.95> <12> <14.4> <17.28> <20.74> <24.88>
      mathx10
      }{}
\DeclareSymbolFont{mathx}{U}{mathx}{m}{n}
\DeclareMathSymbol{\bigtimes}{1}{mathx}{"91}

\usepackage{tcolorbox}
\tcbuselibrary{skins,breakable}
\tcbset{enhanced jigsaw}

\usepackage[normalem]{ulem}
\usepackage[compact]{titlesec}

\definecolor{DarkRed}{rgb}{0.5,0.1,0.1}
\definecolor{DarkBlue}{rgb}{0.1,0.1,0.5}

\usepackage{nameref}
\definecolor{ForestGreen}{rgb}{0.1333,0.5451,0.1333}
\definecolor{Red}{rgb}{0.9,0,0}
\usepackage[linktocpage=true,
	pagebackref=true,colorlinks,
	linkcolor=DarkRed,citecolor=ForestGreen,
	bookmarks,bookmarksopen,bookmarksnumbered]
	{hyperref}
\usepackage[noabbrev,nameinlink]{cleveref}
\crefname{property}{property}{Property}
\creflabelformat{property}{(#1)#2#3}
\crefname{equation}{eq}{Eq}
\creflabelformat{equation}{(#1)#2#3}

\usepackage{bm}
\usepackage{url}
\usepackage{xspace}
\usepackage[mathscr]{euscript}

\usepackage{tikz}
\usetikzlibrary{arrows}
\usetikzlibrary{arrows.meta}
\usetikzlibrary{shapes}
\usetikzlibrary{backgrounds}
\usetikzlibrary{positioning}
\usetikzlibrary{decorations.markings}
\usetikzlibrary{decorations.pathreplacing} 
\usetikzlibrary{patterns}
\usetikzlibrary{calc}
\usetikzlibrary{fit}
\usetikzlibrary{decorations}

\usepackage[framemethod=TikZ]{mdframed}

\usepackage[noend]{algpseudocode}
\makeatletter
\def\BState{\State\hskip-\ALG@thistlm}
\makeatother

\usepackage{cite}
\usepackage{enumitem}
\setlist[itemize]{leftmargin=20pt}
\setlist[enumerate]{leftmargin=20pt}

\usepackage[margin=1in]{geometry}

\usepackage{thmtools}
\usepackage{thm-restate}

\newtheorem{theorem}{Theorem}
\newtheorem{lemma}{Lemma}[section]
\newtheorem{proposition}[lemma]{Proposition}

\newtheorem{fact}[lemma]{Fact}

\newtheorem{definition}[lemma]{Definition}
\newtheorem{problem}{Problem}

\newtheorem*{claim*}{Claim}
\newtheorem*{assumption*}{Assumption}
\newtheorem*{proposition*}{Proposition}
\newtheorem*{lemma*}{Lemma}

\newtheorem*{theorem*}{Theorem}

\crefname{lemma}{Lemma}{Lemmas}
\crefname{claim}{claim}{claims}
\crefname{property}{Property}{Properties}
\crefname{invariant}{Invariant}{Invariants}

\newtheorem{mdresult}{Result}
\newenvironment{result}{\begin{mdframed}[backgroundcolor=lightgray!40,topline=false,rightline=false,leftline=false,bottomline=false,innertopmargin=2pt]\begin{mdresult}}{\end{mdresult}\end{mdframed}}

\theoremstyle{definition}

\newtheorem*{mdproblem*}{Problem}
\newenvironment{Problem*}{\begin{mdframed}[hidealllines=false,innerleftmargin=10pt,backgroundcolor=gray!10,innertopmargin=5pt,innerbottommargin=5pt,roundcorner=10pt]\begin{mdproblem*}}{\end{mdproblem*}\end{mdframed}}
\newtheorem{mddefinition}[lemma]{Definition}

\newtheorem*{mddefinition*}{Definition}
\newenvironment{Definition*}{\begin{mdframed}[hidealllines=false,innerleftmargin=10pt,backgroundcolor=white!10,innertopmargin=5pt,innerbottommargin=5pt,roundcorner=10pt]\begin{mddefinition*}}{\end{mddefinition*}\end{mdframed}}
\newtheorem{mdremark}{Remark}

\newenvironment{ourbox}{\begin{mdframed}[hidealllines=false,innerleftmargin=10pt,backgroundcolor=white!10,innertopmargin=2pt,innerbottommargin=5pt,roundcorner=10pt]}{\end{mdframed}}

\newtheorem{mdalgorithm}{Algorithm}

\allowdisplaybreaks

\renewcommand{\qed}{\nobreak \ifvmode \relax \else
      \ifdim\lastskip<1.5em \hskip-\lastskip
      \hskip1.5em plus0em minus0.5em \fi \nobreak
      \vrule height0.75em width0.5em depth0.25em\fi}

\renewcommand{\leq}{\leqslant}
\renewcommand{\geq}{\geqslant}
\renewcommand{\le}{\leq}
\renewcommand{\ge}{\geq}

\newcommand{\Ot}{\ensuremath{\widetilde{O}}}

\newcommand{\poly}{\mbox{\rm poly}}

\DeclareMathOperator*{\Prob}{\ensuremath{\textnormal{Pr}}}
\renewcommand{\Pr}{\Prob}

\newenvironment{tbox}{\begin{tcolorbox}[
		enlarge top by=5pt,
		enlarge bottom by=5pt,
		 breakable,
		 boxsep=0pt,
                  left=4pt,
                  right=4pt,
                  top=10pt,
                  arc=0pt,
                  boxrule=1pt,toprule=1pt,
                  colback=white
                  ]
	}
{\end{tcolorbox}}

\newcommand{\supp}[1]{\ensuremath{\textnormal{\text{supp}}(#1)}}

\newcommand{\II}{\ensuremath{\mathbb{I}}}

\newcommand{\mireal}[1][]{
  \ifx\relax#1\relax%
    \II(\mione \,; \mitwo)%
  \else%
    \II(\mione \,; \mitwo\mid #1)%
  \fi
}

\title{Deterministic Independent Sets in the Semi-Streaming Model}
\author{Daniel Ye\footnote{(d29ye@uwaterloo.ca) School of Computer Science, University of Waterloo. Supported in part by Sepehr Assadi’s Sloan Research Fellowship and NSERC Discovery grant.}}
\date{February 8, 2025}

\begin{document}

\maketitle

\begin{abstract}
We consider the independent set problem in the semi-streaming model. For any input graph $G=(V, E)$ with $n$ vertices, an independent set is a set of vertices with no edges between any two elements. In the semi-streaming model, $G$ is presented as a stream of edges and any algorithm must use $\Ot(n)$\footnote{$\Ot(\cdot)$ is used to hide polylog factors.} bits of memory to output a large independent set at the end of the stream.

\smallskip
Prior work has designed various semi-streaming algorithms for finding independent sets. Due to the hardness of finding maximum and maximal independent sets in the semi-streaming model, the focus has primarily been on finding independent sets in terms of certain parameters, such as the maximum degree $\Delta$. In particular, there is a simple randomized algorithm that obtains independent sets of size $\frac n{\Delta+1}$ in expectation, which can also be achieved with high probability using more complicated algorithms. For deterministic algorithms, the best bounds are significantly weaker. In fact, the best we currently know is a straightforward algorithm that finds an $\tilde\Omega\left(\frac n{\Delta^2}\right)$ size independent set.

\smallskip
We show that this straightforward algorithm is nearly optimal by proving that any deterministic semi-streaming algorithm can only output an $\Ot\left(\frac n{\Delta^2}\right)$ size independent set. Our result proves a strong separation between the power of deterministic and randomized semi-streaming algorithms for the independent set problem.
\end{abstract}

\newpage
\section{Introduction}
Finding an independent set of a graph is a classical problem in graph theory with wide-ranging applications. An independent set of $G=(V, E)$ is a subset $I\subseteq V$ such that none of the vertices in $I$ have an edge between them. Being able to find large independent sets plays a crucial role in network scheduling, transportation management, and much more. Especially in recent years, there has been an increased need for handling massive graphs to solve various tasks. Consequently, there has been an increased interest in studying the independent set problem in modern models of computation, such as the semi-streaming model introduced in~\cite{feigenbaum2005graph}. In this model, $G$ is presented as a stream of edges and the storage space of the algorithm is bounded to $\Ot(n)$, where $n=|V|$. Our result pertains to the independent set problem under the semi-streaming model. 
\vspace{-0.05cm}

It is known that the Maximum Independent Set problem is both NP-hard\ \cite{karp2010reducibility} and hard-to-approximate to within a factor of $n^{1-\delta}$ for any $\delta>0$\ \cite{zuckerman2006linear}. Not only that, it is hard-to-approximate in the semi-streaming model, regardless of the time complexity of the algorithm\ \cite{halldorsson2012streaming}. As such, one line of research has focused on obtaining \textit{Maximal} Independent Set. In this respect,~\cite{ahn2015correlation} yields an $O(\log \log n)$-pass semi-streaming algorithm for this problem, and very recently,\ \cite{assadi2024log} proved this is the optimal number of passes.
\vspace{-0.05cm}

Hence, to study algorithms in even fewer passes, the problem must be further relaxed to finding ``combinatorially optimal'' independent sets. On one hand, we can try finding an optimal bound with respect to degree sequences --- the Caro-Wei theorem guarantees an independent set of size $\sum_{v}\frac 1{1+\deg(v)}$, which is optimal in the sense that there are many graphs that do not admit larger independent sets. To achieve this bound in the semi-streaming model,~\cite{halldorsson2010streaming} devises an algorithm for hypergraphs, which when applied to graphs, yields an expected independent set of size $\Omega\left(\sum_{v}\frac 1{\deg(v)}\right)$ with $O(n)$ memory and $O(1)$ time per edge. 
\vspace{-0.05cm}

On the other hand, if we consider bounds with respect to $n$ and $\Delta$ (the maximum degree in a graph), the Caro-Wei theorem guarantees independent sets of size $\frac{n}{\Delta+1}$, which is also tight. There is a very easy randomized semi-streaming algorithm for achieving this bound in expectation: randomly permute the vertices, and when an edge arrives in the stream, mark the endpoint that appears later. At the end, output all unmarked vertices. It is a standard result in graph theory that this algorithm outputs an independent set of size $\frac n{\Delta+1}$. We can even obtain such an independent set with high probability using the $(\Delta+1)$-coloring semi-streaming algorithm presented in~\cite{assadi2019sublinear}.
\vspace{-0.05cm}

Interestingly, these results all make heavy use of randomization. Thus, as with a plethora of other problems in the semi-streaming model, there is a particular interest in derandomizing algorithms to achieve similar performance. Some problems admit single-pass deterministic algorithms matching their randomized counterparts (e.g.\ connectivity and bipartiteness~\cite{ahn2012analyzing}, maximum matching~\cite{paz20182+}). However, others have a proven discrepancy between the theoretical space complexities of deterministic algorithms and randomized ones under this model (e.g.\ triangle counting~\cite{braverman2013hard} and vertex coloring~\cite{assadi2022deterministic}). 
\vspace{-0.05cm}

For the independent set problem, the state-of-the-art has been a fairly simple deterministic algorithm for finding an independent set of size $\frac n{\Delta^2}$. We begin by choosing $\frac n{\Delta}$ vertices arbitrarily, then store all the edges between them (resulting in $O\left(\frac n{\Delta}\cdot \Delta\right)=O(n)$ edges). We can then calculate a maximal independent set of this subgraph, resulting in an independent set of size $O\left(\frac n{\Delta^2}\right)$. There has been very little progress in constructing a deterministic algorithm that does better than this, which has led to the following open question:
\vspace{-0.3cm}
\begin{quote}
    \begin{center}
        \textit{Is there a single-pass semi-streaming deterministic algorithm that can match the performance of randomized ones? In particular, is there a deterministic algorithm that can find an independent set of size better than $\Ot\left(\frac n{\Delta^2}\right)$ in general graphs?}
    \end{center}
\end{quote}

\subsection{Our Contributions}

Our main result is a negative answer to the open question: the deterministic algorithm stated above is optimal up to polylog factors.

\begin{result}\label{main-result}

Any deterministic single-pass semi-streaming algorithm can only find an independent set of size at most $\Ot\left(\frac n{\Delta^2}\right)$ in a graph of maximum degree $\Delta$ (even when $\Delta$ is known). 
\end{result}

To the best of our knowledge, no space lower bounds have been devised for deterministic algorithms solving the independent set problem in the semi-streaming model. This result provides a lower bound that is tight up to logarithmic factors, illustrating another significant separation between deterministic and randomized algorithms in the semi-streaming model. 

\subsection{Our Techniques}\label{sec:our-technique}
We give a summary of our techniques here.

We begin our proof of the space lower bound in \Cref{main-result} with a similar setup to~\cite{assadi2022deterministic}, who derived a space lower bound for deterministic algorithms solving the vertex coloring problem. In our case, we consider the multi-party communication complexity of the Independent Set problem, where the edges of a graph are partitioned among some number of players. In a predefined order, each player may speak once (outputting $\Ot(n)$ bits) and is heard by all future players. The goal is for the last player to output a large independent set. 

We will design an adversary that adaptively constructs a graph that forces the algorithm to output an independent set of size $\Ot\left(\frac n{\Delta^2}\right)$. To do so, it is useful to consider the graph of non-edges (later referred to as the missing graph): the graph consisting of edges each player \textit{knows} has not been sent to them or any previous player. This model is useful because the independent set the last player outputs \textit{must} be a clique in their missing graph (otherwise, there would be some input that makes this algorithm incorrect). 

Additionally, we make use of the compression lemma of~\cite{assadi2022deterministic}. On any arbitrary graph, if we sample each edge with probability $p$, for any algorithm that compresses the result into $s$ bits, the compression lemma finds a summary such that at most $O\left(\frac sp\right)$ edges are not in any graph mapped to that summary. This is useful for\ \cite{assadi2022deterministic} as their adversary can narrow its search to a small set of vertices that a deterministic algorithm cannot summarize well. In fact, by focusing on vertices with low non-edge degree, it can sample each remaining edge with a higher probability to improve the bound given by the compression lemma. However, for the independent set problem, deterministic algorithms are free to choose vertices from \textit{any} section of the graph (and do not need to deal with \textit{every} vertex as in coloring). As such, our adversary does not have the luxury of searching for some useful set of vertices nor working with vertices with low non-edge degree. Instead, it must remove large independent sets globally from the graph and account for vertices that might not be easy to work with. To achieve this, our adversary generates graphs that are similar in structure to a Tur\'an graph. In particular, as a deterministic algorithm receives its edges, the overall graph will seem more like many densely-connected vertex-disjoint subgraphs. The key to this strategy is a new lemma in our paper that allows for destroying large independent sets (equivalently, cliques) by adding ``few'' edges (equivalently, removing in the case of cliques).

\subsection{Related Works}

A similar-in-spirit result for vertex coloring is proven in~\cite{assadi2022deterministic}, which shows that deterministic algorithms cannot color a graph with $\exp(\Delta^{o(1)})$ colors in a single pass. Since vertex coloring is fairly hard, designing an adversary entails finding \textit{some} set of vertices that is a clique from the perspective of a deterministic algorithm. With the independent set problem, however, we must design an adversary that removes \textit{all} large independent sets from the perspective of a deterministic algorithm. Due to this difference, deterministic algorithms can easily find independent sets of size $\frac n{\poly(\Delta)}$ in our setting (whereas no deterministic algorithm can find a coloring using at most $\poly(\Delta)$ colors). We show that they cannot do better than quadratic, up to a polylog factor. 

More generally, there has also been a significant interest in finding independent sets in graph streams\ \cite{ahn2015correlation, assadi2024log, halldorsson2010streaming, chen2023sublinear, cormode2017independent, bhore2022streaming, bakshi2019weighted}. Independent sets in the \textit{online streaming} model are studied in~\cite{halldorsson2016streaming}. Under this model, they devise a deterministic algorithm with performance ratio $O(2^{\Delta})$, which they prove is also tight. We provide an adjacent result in the \textit{semi-streaming} model, which does not require an algorithm to maintain a feasible solution at all times. Additionally,\ \cite{cormode2018independent} studied independent sets in vertex-arrival streams, where each element in the input stream is a vertex along with its incident edges to earlier vertices. They show that the maximum independent set problem in the vertex-arrival model is not much easier than the problem in the edge-streaming model. 

Independent sets have also been studied in more tangential settings.\ \cite{cormode2018approximating} studied the problem of finding the Caro-Wei bound itself that other algorithms (such as the one in\ \cite{halldorsson2010streaming}) achieve, and\ \cite{bhore2022streaming} studied the geometric independent set problem in the streaming model.

\section{Preliminaries}
\newcommand{\Partition}{\textnormal{\textsc{Partition}}}
\newcommand{\CP}{\mathcal{P}}

\textbf{Notation. } For an integer $k\ge 1$, we denote $[k]:=\{1,2,\dots,k\}$. For a tuple $(X_1, \dots, X_k)$ and any $i\in [k]$, we denote $X_{<i}$ as $(X_1, \dots, X_{i-1})$. For any distribution $\mu$, we will denote $\supp{\mu}$ as the support of $\mu$. 

For a graph $G=(V, E)$, we denote $\Delta{(G)}$ as the maximum degree of $G$, and for any $v\in V$, $\deg{(v)}$ as the degree of $v$ in $G$. For any vertex set $T\subseteq V$, we will denote the induced subgraph of $G$ on $T$ as $G[T]$. Often, we will also partition a set $S$ into a collection of subsets $\CP$ and a subset $Q$. When we use this language, we are saying that $\CP\cup \{Q\}$ is a partition of $S$. 

One part of our strategy also involves partitioning the vertices into many small subsets. For any integer $g\ge 1$, we will denote $\Partition(S, g)$ as an arbitrary partition of $S$ into subsets of size $g$, \textit{except} potentially for the last set, which has size $<g$. 
\begin{fact}\label{partition-count}
    For any set $S$ and $g\ge 1$, $|\Partition(S, g)|\le \left\lceil\frac{|S|}g\right\rceil$.
\end{fact}

Finally, we use the following standard Chernoff bound:

\begin{proposition}[Chernoff bound; c.f.~\cite{dubhashi2009concentration}]\label{chernoff}
    Suppose $X_1, \dots, X_m$ are $m$ independent random variables in the range $[0, 1]$. Let $X:=\sum_{i=1}^m X_i$ and $\mu_L\le \mathbb{E}[X]\le \mu_H$. Then, for any $\epsilon > 0$,
    \begin{align*}
        \Pr(X > (1 + \epsilon)\cdot \mu_H)\le\exp\left(-\frac{\epsilon^2\cdot\mu_H}{3+\epsilon}\right)~\mathrm{and}~\Pr\left(X < (1-\epsilon)\cdot \mu_L\right)\le\exp\left(-\frac{\epsilon^2\cdot\mu_L}{2+\epsilon}\right).
    \end{align*}
\end{proposition}

\subsection{The Communication Complexity of Independent Sets}
We prove our space lower bound in \Cref{main-result} through a communication complexity argument in the following communication game, which as stated in \Cref{sec:our-technique}, is defined similarly to~\cite{assadi2022deterministic}.

\newcommand{\IndepSet}{\textnormal{\textsc{Independent-Set}}}

\begin{mdframed}
    For integers $n, \Delta, k, s\ge 1$, the $\IndepSet(n, \Delta, k, s)$ game is defined as:
    \begin{enumerate}
        \item There are $k$ players $P_1, \dots, P_k$. Each Player $P_i$ knows the vertex set $V$ and receives a set $E_i$ of edges. Let $G=(V, E)$, where $E=E_1\cup \dots\cup E_k$ and players are guaranteed $E_1,\dots, E_k$ are disjoint. Players are guaranteed $\Delta(G)\le \Delta$, and their goal is to output an Independent set of $G$.
        \item In order from $i=1$ to $i=k$, each player $P_i$ writes a public message $M_i$ based on $E_i$ and $M_{<i}$ (all the messages from the previous players) of length at most $s$. 
        \item The goal of the players is to output an independent set of $G$ by $P_k$ outputting it as the message $M_k$. 
    \end{enumerate}
\end{mdframed}

The following is standard:

\begin{lemma}\label{algo-to-players}
    Suppose there is a deterministic streaming algorithm that, on any $n$-vertex graph $G$ with known maximum degree $\Delta$, outputs an independent set of $G$ with size $r$ using $s$ bits of space. Then, there also exists a deterministic protocol for $\IndepSet(n,\Delta, k, s)$ that outputs an independent set of size $r$.
\end{lemma}
\begin{proof}
The players simply run the streaming algorithm on their input by writing the content of the memory of the algorithm from one player to the next on the blackboard, so that the next player can continue running the algorithm on their input. At the end, the last player computes the output of the streaming algorithm and writes it on the blackboard.
\end{proof}

\subsection{The Missing Graph and Compression Lemma}

As stated in \Cref{sec:our-technique}, we use the compression lemma in~\cite{assadi2022deterministic}. We begin with two definitions that are similar to~\cite{assadi2022deterministic}.

\begin{definition}
    For a base graph $G_{Base}=(V, E_{Base})$ and parameters $p\in (0, 1]$, $d\ge 1$, we define the random graph distribution $\mathbb{G}=\mathbb{G}(G_{Base}, p, d)$ as follows:
    \begin{enumerate}
        \item Sample a graph $G$ on vertices $V$ and edges $E$ by picking each edge of $E_{Base}$ independently with probability $p$. 
        \item Return $G$ if $\Delta{(G)}< 2p\cdot d$. Otherwise, repeat the process.
    \end{enumerate}
\end{definition}

To analyze arbitrary deterministic algorithms, we will often consider the compression algorithm associated with it. We will represent the ``information'' available to the algorithm as a ``missing graph'', which we define here:
\begin{definition}
    Consider $\mathbb{G}(G_{Base}, p, d)$ for a base graph $G_{Base}=(V, E_{Base})$ and parameters $p\in(0, 1]$, $d\ge 1$, and an integer $s\ge 1$. A compression algorithm with size s is any function $\Phi : \supp{\mathbb{G}} \rightarrow {\{0, 1\}}^s$ that maps graphs sampled from $\mathbb{G}$ into $s$-bit strings. For any graph $G\in \supp{\mathbb{G}}$, we refer to $\Phi(G)$ as the summary of G. For any summary $\Phi\in{\{0,1\}}^s$, we define:
    \begin{enumerate}
        \item $\mathbb{G}_{\phi}$ as the distribution of graphs mapped to $\phi$ by $\Phi$. 
        \item $G_{Miss}(\phi)=(V, E_{Miss}(\phi))$, called the missing graph of $\phi$, as the graph on vertices $V$ and edges missed by all graphs in $\mathbb{G}_{\phi}$. 
    \end{enumerate}
\end{definition}

The previous definitions are used extensively in our work. The following lemma also plays a crucial role in our communication lower bound, bounding the number of conclusive missing edges that can be recovered from a compression algorithm of a given size. It is proven in~\cite{assadi2022deterministic}.

\begin{lemma}[Compression Lemma]\label{compression}
Let $G_{Base}=(V, E_{Base})$ be an $n$-vertex graph, $s\ge 1$ be an integer, and $p\in (0, 1)$ and $d\ge 1$ be parameters such that $d\ge \max\{\Delta{(G_{Base})}, \frac {4\ln(2n)}p\}$. Consider the distribution $\mathbb{G}:=\mathbb{G}(G_{Base}, p, d)$ and suppose $\Phi:\supp{\mathbb{G}}\rightarrow {\{0, 1\}}^s$ is a compression algorithm of size $s$ for $\mathbb{G}$. Then, there exists a summary $\phi^*\in {\{0, 1\}}^s$ such that in the missing graph of $\phi^*$,
\begin{align*}
    |E_{Miss}(\phi^*)|\le \frac{\ln 2\cdot (s+1)}p.
\end{align*}
\end{lemma}
\section{Removing Cliques in the Missing Graph}

\indent In this section, we introduce a key tool used by our adversary. More specifically, we develop a procedure that removes large cliques in the missing graph: the graph consisting of edges that an algorithm \textit{knows} are not in the input graph. This is useful because a deterministic algorithm that finds independent sets must be \textit{certain} that none of the vertices in its output have an edge between them, which creates a clique of the same size in its missing graph. To do so, we will design our adversary to generate a Tur\'an-type graph, as mentioned in \Cref{sec:our-technique}.

\subsection{Removing Cliques in Low-Degree Graphs}
The following lemma provides a method to bound the largest size of a clique when the degree is bounded by removing a small number of edges.

\begin{lemma}\label{clique-removal}
Let $G$ be a graph with $n$ vertices such that $\Delta{(G)}\le \Delta$. For any positive integer $d$, there is some subgraph $H$ of $G$ such that:
\begin{enumerate}
    \item\label{S-degree} The degree of $H$ is $\le d$.
    \item\label{clique-size} The largest clique in $G-H$ has size $\le 16\ln(n)\cdot \frac {\Delta}d + 10$.
\end{enumerate}
\end{lemma}
\begin{proof}
    We will use a probabilistic argument.
    
    Firstly, if $d\le 16\ln (n)$, then we let $H$ be the empty graph. Since $\Delta{(G)}\le \Delta$, the largest clique in $G$ has size $\le \Delta$. Indeed, $16\ln (n)\cdot \frac {\Delta}d+10\ge 16\ln(n)\cdot \frac \Delta{16\ln(n)}=\Delta$. 

    Similarly, if $\Delta\le d$, then take $H=G$. The largest clique size is $1$ and $\Delta{(H)}\le \Delta\le d$. Finally, if $n\le 10$, then the largest clique size is at most $10$, so we can let $S$ be the empty graph as well. 
    
    It remains to prove the lemma for $n>10$ and $16\ln(n)<d<\Delta$. We choose $S$ by sampling each edge in $G$ with probability $\frac {d}{2\Delta}$. For any vertex $v$, we let $X_1,\dots,X_{\deg{(v)}}$ be indicator variables for whether an incident edge is chosen. If we let $X:=\sum_{i=1}^{\deg{(v)}} X_i$, then $\mathbb{E}[X]\le \frac {d}2$. By \Cref{chernoff} with $\epsilon=1$, 
    \begin{align*}
        \Pr(X>d)\le \Pr\left(X>2\cdot\mathbb{E}[X]\right)\le\exp\left(-\frac{\frac{d}{2}}{4}\right)<\exp\left(-\frac {16\ln n}{8}\right) =\exp(-2\ln n).
    \end{align*}

    Hence, by a union bound, $\Pr(\text{Some vertex in $H$ has degree greater than $d$})$ does not exceed
    \begin{align*}
        n\cdot \exp\left(-2\ln n\right)=\exp(\ln n-2\ln n)=\exp (-\ln n)=\frac 1n.
    \end{align*}
    
    For any group of $16\ln(n)\cdot \frac {\Delta}d+10<k$ vertices, if it is not a clique in $G$, it definitely will not be a clique after removing the edges in $H$. Otherwise, it can only be a clique if none of the edges are selected. For any $k$-subset, the probability of this happening is:
    \begin{align*}
        \Pr(\text{this $k$-subset is a clique})&={\left(1-\frac {d}{2\Delta}\right)}^{k\cdot(k-1)/2}\\
        &\le\exp\left(-\frac d{2\Delta}\cdot \frac{k(k-1)}2\right)\\
        &<\exp\left(-\frac{d}{2\Delta}\cdot \left(16\ln n\cdot \frac {\Delta}d\right)\cdot \frac{k-1}2\right)\\
        &=\exp\left(-4\ln n\cdot (k-1)\right).
    \end{align*}

    Now, we apply a union bound over all k-subsets of vertices:
    \begin{align*}
        \Pr(G-H\text{ has a k-clique})&\le \binom nk\cdot \exp\left(-4\ln n\cdot (k-1)\right)\\
        &\le {\left(\frac{en}k\right)}^k \cdot \exp\left(-4\ln n\cdot (k-1)\right)\\
        &=\exp(k\ln(en/k)-4\ln(n)\cdot (k-1))\\
        &=\exp(k(1+\ln(n)-\ln(k)-4\ln(n)) + 4\ln n)\\
        &\le \exp(k(1-3\ln(n))+4\ln(n))\\
        &\le \exp(-2k\ln n + 4\ln(n))\hspace{5em}\tag{Since $\ln(n)\ge1$}\\
        &\le \exp\left(-6\ln(n)+4\ln(n)\right)\hspace{4em}\tag{Sub $k=16\ln(n)\cdot \frac{\Delta}d\ge3$}\\
        &=\frac 1{n^2}.
    \end{align*}

    Finally, we calculate the probability that this procedure yields a subgraph $H$ satisfying \Cref{S-degree} and \Cref{clique-size}. Through a union bound on the complement, when $n>10$, we get
    \begin{align*}
        \Pr(H\text{ does not satisfy \Cref{S-degree} or \Cref{clique-size}})\le \frac1{10}+\frac1{100}<1.
    \end{align*} 
    Thus, there is some $H$ satisfying the two conditions. 
\end{proof}

\subsection{Removing Cliques in General Graphs}
\noindent
To use \Cref{clique-removal}, we need graphs of low maximum degree. However, the $\hyperref[compression]{\text{Compression Lemma}}$ can only bound the total number of edges. Hence, to employ \Cref{clique-removal}, we will partition a graph with few edges into many subgraphs with low degree and a small ``remainder'' subgraph. We will start with the following definition:

\newcommand{\graphSplit}{\textnormal{\textsc{Split}}}

\begin{definition}
    For a positive integer $b$ and a graph $G=(V, E)$ with $m$ edges, we define $\graphSplit(G, b)$ as a partition of $V$ into a pair of vertex sets $(P, Q)$ such that $\Delta{(G[P])}\le b$ and $|Q|\le \frac{2m}b$.
\end{definition}

\noindent
The following proposition ensures that a split exists:

\begin{proposition}\label{partition-degree-bound}
Suppose we are given a graph $G=(V, E)$ with $n$ vertices and $m$ edges. For every positive integer $b$, $\graphSplit(G, b)$ exists.
\end{proposition}
\begin{proof}
Let $P:=\{v\in V:\deg{(v)}\le b\}$. Then, $Q:=V\backslash P$. This is clearly a partition, since each vertex is in exactly one of $P$ or $Q$. 

We will now prove the bound on $G[P]$. For every vertex $v\in P$, $\deg{(v)}\le b$ in $G$. Since $G[P]$ is a subgraph of $G$, $\deg{(v)}\le b$ in $G[P]$ as well. Hence, $\Delta{(G[P])}\le b$. 

Now for the bound on $|Q|$. It is known that the number of edges is equal to half the sum of the degrees. Hence, we have 
\begin{align*}
    m=\frac12\sum_{v\in V}\deg{(v)}\ge \frac12\sum_{v\in V\backslash P}\deg{(v)}\ge \frac 12\sum_{v\in Q} b=\frac 12\cdot |Q|\cdot b.
\end{align*}
Then, we rearrange to get $|Q|\le \frac{2m}b$. 
\end{proof}

To ``remove'' many large cliques in an arbitrary graph, we will run a two-step subalgorithm. 

The first step of the subalgorithm involves finding a subgraph of our input graph $G_{Base}$, which we denote as $H_{Base}$. The graph $H_{Base}$ is chosen such that, for any algorithm compressing it, we can easily bound the number of edges in $H_{Miss}(\phi)$ (the missing graph of $H_{Base}$ for some message $\phi$) using the \hyperref[compression]{Compression Lemma}. We will then prove that we can partition the vertices of $H_{Miss}(\phi)$ (which are the same as the vertices of $G_{Base}$) into a collection of vertex sets $\CP$ and a vertex set $Q$ such that the maximum degree in $\left(H_{Miss}(\phi)\right)[P]$ is small for all $P\in \CP$ and the size of $Q$ is small. 

\begin{lemma}\label{compressed-partition}
Let $G_{Base}=(V, E_{Base})$ be a graph with $n$ vertices. Let $g\ge 1$ (group size) and $s\ge 1$ (message size) be arbitrary integers. Let $d_{comp}\ge4\ln(2n)$ (compression degree) and $d_{filter}\ge 1$ (filter degree) be arbitrary real numbers.

Then, there is a subgraph $H_{Base}\subseteq G_{Base}$ and a distribution $\mathbb{H}:=\mathbb{G}(H_{Base}, p, d)$ (for some $p$ and $d$) such that, for every compression algorithm $\Phi:\supp{\mathbb{H}}\rightarrow {\{0, 1 \}}^s$, we can find a message $\phi$ and a partition of $V$ into a collection of vertex sets $\CP$ and a vertex set $Q$ satisfying:
\begin{enumerate}
    \item\label{H-degree} For all $H\in \supp{\mathbb{H}}$, $\Delta{(H)}\le 2\cdot d_{comp}$.
    \item\label{edge-in} For all $P\in \CP$, $G_{Base}[P]=H_{Base}[P]$.
    \item\label{P-condition} For any vertex set $P\in \CP$, $\Delta{(H_{Miss}(\phi)[P])}\le d_{filter}$. 
    \item\label{P-size} The size of $\CP$ is $\le\lceil\frac ng\rceil$. 
    \item\label{Q-condition} The size of $Q$ is $\le \dfrac{2\ln (2)\cdot (s+1)\cdot g}{d_{comp}\cdot d_{filter}}$.
\end{enumerate}
\end{lemma}
\begin{proof}
Let $H:=\Partition(G_{Base}, g)$ and $H_{Base}:=\bigcup_{S\in H}G_{Base}[S]$. We prove this lemma through casework on $d_{comp}$.\\

If $d_{comp}<g$, then we define $\mathbb{H}:=\mathbb{G}\left(H_{Base}, \frac {d_{comp}}g, g\right)$. 

For all $H\in \supp{\mathbb{H}}$, $\Delta{(H)}\le 2\cdot \frac {d_{comp}}g\cdot g=2\cdot d_{comp}$ by construction, proving \Cref{H-degree}. 

These values of $p=\frac {d_{comp}}g$ and $d=g$ satisfy the requirements for the $\hyperref[compression]{\text{Compression Lemma}}$. In particular, since $4\ln(2n)\le d_{comp}<g$, we have $p=\frac {d_{comp}}g\in (0, 1)$. Additionally, $d=g\ge \Delta{(H_{Base})}$ since each subgraph has at most $g$ vertices, and $d=g= \frac {d_{comp}}{d_{comp}/g}\ge \frac {4\ln (2n)}p$. Hence, the values of $p$ and $d$ satisfy the necessary constraints. 

Therefore, there exists some message $\phi$ such that $H_{Miss}(\phi)=(V, E_{Miss}(\phi))$ satisfies 
\begin{align*}
    |E_{Miss}(\phi)|\le \frac{\ln(2)\cdot (s+1)\cdot g}{d_{comp}}.
\end{align*}

Let $(L, Q):=\graphSplit(H_{Miss}(\phi), d_{filter})$. By \Cref{partition-degree-bound}, $|Q|\le \frac{2\ln(2)\cdot (s+1)g}{d_{comp}\cdot d_{filter}}$, proving \Cref{Q-condition}. Additionally, let  $\CP:=\{S\cap L:S\in H\}$. Since $H$ is a partition of $V$ and $Q=V\backslash L$, we know that $\CP\cup \{Q\}$ is a partition of $V$. By \Cref{partition-count}, $|\CP|=|H|\le \lceil\frac ng\rceil$, proving \Cref{P-size}.

We will now prove the conditions on all $P\in \CP$. Let $X_P:=H_{Miss}(\phi)[P]$. 
\begin{itemize}
    \item By \Cref{partition-degree-bound}, since $P\subseteq L$, $\Delta{(X_P)}\le \Delta{(H_{Miss}(\phi)[L])}\le d_{filter}$. This proves \Cref{P-condition}.
    \item We know $P\subseteq S$ for some $S\in H$. Additionally, $H$ is a partition of $V$. So, by our construction of $H_{Base}$, we have $H_{Base}[S]=G_{Base}[S]$. Hence, since $P\subseteq S$, we have $H_{Base}[P]=G_{Base}[P]$. This proves \Cref{edge-in}.
\end{itemize}

If $d_{comp}\ge g$, then we let $\mathbb{H}=\Phi(H_{Base}, 1, |V|)$. We let $\CP=H$, so $Q=\emptyset$, which satisfies \Cref{Q-condition}. The distribution $\mathbb{H}$ is a single graph $H_{Base}$. Since each subgraph in $H_{Base}$ has $\le g$ vertices, $\Delta{(H_{Base})}\le g\le 2\cdot d_{comp}$, proving \Cref{H-degree}. Since $|\CP|=|H|$, \Cref{P-size} is proven by \Cref{partition-count}. Since $\supp{\mathbb{H}}=\{H_{Base}\}$, $H_{Miss}(\phi)=\emptyset$, proving \Cref{P-condition}. By how we defined $\CP$ and $H_{Base}$ (note that $H$ is a partition), $G_{Base}[P]=H_{Base}[P]$ for all $P\in \CP$. This proves \Cref{edge-in}. 
\end{proof}

In the second step of the subalgorithm, we will apply \Cref{clique-removal} on $H_{Miss}(\phi)[P]$ for all $P\in \CP$, storing the removed edges in a subgraph $R$ (which will have low degree). In the end, the maximum clique size in $H_{Miss}(\phi)-R$ will be small, since it cannot exceed the \textit{sum} of the maximum clique sizes over the subgraphs induced by the vertex sets in $\CP$.

\begin{lemma}\label{algo-second-step}
Suppose we have a graph $H$ with $n$ vertices and a partition of its vertices into a collection of vertex sets $\CP$ and a vertex set $Q$. Additionally, for all $P\in \CP$, suppose that the degree of $G[P]$ does not exceed $d_{filter}$.

Then, for any integer $d_{remove}>0$ (removal degree) there is a subgraph $R\subseteq H$ such that: 
\begin{itemize}
\item For all $P\in \CP$, the largest clique in $(H_{Miss}(\phi)-R)[P]$ has size $\le 16\ln(n)\cdot \frac{d_{filter}}{d_{remove}}+10$. 
\item None of the edges in $R$ is incident to a vertex in $Q$. 
\item The degree of $R$ does not exceed $d_{remove}$. 
\end{itemize}
\end{lemma}
\begin{proof}
For all $P\in \CP$, by \Cref{clique-removal} on $H[P]$ and $d_{remove}$, there is a graph $R_P$ with degree $\le d_{remove}$ s.t.\ the maximum clique size in $H[P]-R_P$ is $\le 16\ln(n)\cdot\frac{d_{filter}}{d_{remove}}+10$. 

We take $R:=\bigcup_{P\in \CP}R_P$. Since $Q$ is disjoint with all sets in $\CP$, none of the edges in $R$ are incident to any vertex in $Q$. Additionally, for all $P\in \CP$, since $R_P\subseteq R$ and the vertex sets in $\CP$ are disjoint, we have $R_P=R[P]$. Thus, the largest clique in $(H-R)[P]=H[P]-R_P$ has size $\le 16\ln(n)\cdot\frac{d_{filter}}{d_{remove}}+10$.
\end{proof}

\section{A Communication Lower Bound for Independent Set}

We will prove our lower bound in \Cref{main-result} by showing that it is impossible for any set of players to output a large independent set, which is sufficient by \Cref{algo-to-players}. To do so, we will design an adversary that, for any large independent set, can find an \textit{invalid} graph and set of edges to send to each player such that each player outputs the same message. This will ensure that no deterministic algorithm can confidently output any large independent set. 

\subsection{The Adversary}

\begin{theorem}\label{main-theorem}
    There exist constants $\eta>0$ and $\eta_0>0$ such that: if $n$, $s$, and $\Delta$ are integers satisfying
    \[ \quad \eta<n\le s, \quad \max\left\{\eta_0\cdot \frac{s\ln(n)}n, \eta_0\cdot {(\ln n)}^2 \right \}< \Delta < \sqrt n,\]
    then the size of the largest independent set a deterministic algorithm using $s$ bits of memory can output for all graphs of size $n$ and maximum degree $\Delta$ does not exceed $\Ot\left(\frac s{\Delta^2}\cdot \frac sn\right)$. 
\end{theorem}

Under the semi-streaming model, $s=\Ot(n)$. Consequently, the largest independent set that a deterministic algorithm can find under this model has size $\Ot\left(\frac n{\Delta^2}\right)$, which formalizes \Cref{main-result}.

\newcommand{\distribFac}{\ell}

To begin with the proof of \Cref{main-theorem}, we will let $\eta=e^2$ and $\eta_0=128$. We let ${\distribFac}=\max\left\{\left\lceil\frac{2e\ln(2)(s+1)}n \right\rceil, \left\lceil8\ln n\right\rceil\right\}$. By our choice of $\eta$ and $\eta_0$, it is easy to prove that ${\distribFac}<\frac{\Delta}{4\ln(2n)}$. We let $k=\lceil \ln n\rceil +1$. For our adversary, $k$ denotes the number of players, and we will generally send graphs of degree $\le \frac{\Delta}{\distribFac}$ to each player. 

At a high level, our adversary adheres to the following structure:
\begin{itemize}
    \item For all $i=1\dots k$, $G_{Base}(i)$ is the graph where most edges for Player $i$ will be chosen from.
    \item Similarly, $R_i$ will be provided as an additional set of edges to send to Player $i$ to ``manually'' remove large cliques in the missing graph for the previous players' messages. 
    \item For all $i\in [k]$, Player $i$ receives $R_i$ and a subgraph $H_i\subseteq G_{Base}(i)$. 
    \item \Cref{compressed-partition} is used to find an adversarial distribution of subgraphs $\mathbb{H}_i$, which is used to obtain $H_i$. \Cref{algo-second-step} is used to determine $R_{i+1}$, which is a subgraph of $G_{Base}(i)$.
    \item \textit{All} graph parameters are derived adversarily based on what each player communicates. The compression algorithm associated with each player also affects the input for future players by determining $G_{Base}(i+1)$ and $R_{i+1}$.
    \item The adversarily generated input graph, $G_{input}$, is the union of all graphs $H_i$ and $R_i$. 
\end{itemize}

\begin{mdframed}
An \textbf{adversary} that generates a ``hard'' input. 

\begin{enumerate}
    \item\label{1} We start with $G_{Base}(1)$ as the clique on $n$ vertices and $R_1$ as an empty graph.
    \item For $i=1\dots k$,
        \begin{enumerate}
            \item Let $n_i$ be the number of vertices in $G_{Base}(i)$, where $G_{Base}(i)=(V_{Base}(i), E_{Base}(i))$. 
            \item If $n_i<\frac n{\Delta^2}$, then terminate the algorithm, letting $\mathbb{H}=\emptyset$ and $R_{i+1}=\emptyset$. For the sake of our analysis, we let $G_{Base}(i+1)=G_{Base}(i)$ and run the adversary until $i=k$. 
            \item Otherwise, apply \Cref{compressed-partition} with $G_{Base}(i)$, group size $\lfloor\frac {n_i\Delta^2}n\rfloor$, message size $s$, compression degree $\frac {\Delta}{\distribFac}$, and filter degree ${\distribFac}^2\Delta$. 
            
            We have a distribution $\mathbb{H}_i$ with a base graph $H_{Base}(i)$, and we define $\Phi_i=\Phi(\mathbb{H}_i, M_{<i})$ as the compression algorithm for Player $i$ \textit{after} receiving $R_i$. 
            
            By \Cref{compressed-partition}, there is a message $M_i:=\phi$ and some partition of $V_{Base}(i)$ into a collection of vertex sets $\CP_i$ and a vertex set $Q_i$ satisfying the conclusions of \Cref{compressed-partition}.

            \item Apply \Cref{algo-second-step} with $H_{Miss}(M_i)$ and removal degree $\frac{\Delta}2$ to find some $R_{i+1}$ such that the conclusions of \Cref{algo-second-step} holds. 
            
            \item Let $G_{Base}(i+1)=(G_{Base}(i)-(H_{Base}(i)-H_{Miss}(M_i)))[Q_i]$. 
        \end{enumerate}
        \item Finally, we generate the edges we send each player. For $i=1,\dots,k$,
        \begin{enumerate}
            \item If $\mathbb{H}=\emptyset$, let $H_i=\emptyset$. Otherwise, choose $H_i\in \supp{\mathbb{H}_i}$ such that $\Phi_i(H_i)=M_i$.
            \item We send Player $i$ the graph $H_i\cup R_i$. 
        \end{enumerate}
        Then, the input graph is the union of these graphs. In particular, $G_{input}:=\bigcup_{i=1}^k H_i\cup R_i$. 
\end{enumerate}
\end{mdframed}

To begin, we must ensure that the input graph $G_{input}$ is valid. The following lemma ensures that we are not sending a multigraph to the players. 

\begin{lemma}\label{graph-is-not-multi}
    No two players will receive the same edge (i.e. $G_{input}$ is not a multigraph).
\end{lemma}
\begin{proof}
It is sufficient to prove the stronger statement that $\{H_i:i\in[k]\}\cup \{R_{i+1}:i\in[k]\}$ is disjoint. For each $i\in [k]$, by \Cref{compressed-partition} and \Cref{algo-second-step}, $H_{Base}(i)\subseteq G_{Base}(i)$ and $R_{i+1}\subseteq H_{Miss}(M_i)\subseteq H_{Base}(i)\subseteq G_{Base}(i)$. Hence, both $H_i$ and $R_{i+1}$ are subgraphs of $G_{Base}(i)$. 

Firstly, we claim that $H_i$ and $R_{i+1}$ are edge-disjoint. In particular, since $H_i$ is sampled from $\supp{\mathbb{H}_i}$ and $\Phi_i(H_i)=M_i$, by definition, we know that none of the edges in $H_i$ are in $H_{Miss}(M_i)$. However, $R_{i+1}\subseteq H_{Miss}(M_i)$, so none of the edges in $R_{i+1}$ are in $H_i$.

Next, we prove that $G_{Base}(i+1)$ is edge-disjoint with $H_i$ and $R_{i+1}$. For any edge in $H_i$, it is in $H_{Base}(i)$ but not $H_{Miss}(i)$. Hence, it is not in $G_{Base}(i)-(H_{Base}(i)-H_{Miss}(M_i))$. For $R_{i+1}$, by \Cref{algo-second-step}, none of its edges are incident to a vertex in $Q_i$. Thus, none of the edges in $R_{i+1}$ appear in $G_{Base}(i+1)=(G_{Base}(i)-(H_{Base}(i)-H_{Miss}(M_i)))[Q_i]$.

Finally, for any $j>i$, we know that $G_{Base}(j)\subseteq \dots\subseteq G_{Base}(i+1)$. Thus, for any edge $(u, v)$ in $H_i$ or $R_{i+1}$, we know that $(u, v)$ is not in $G_{Base}(j)$ (since $(u, v)\not\in G_{Base}(i+1)$). Since $H_j$ and $R_{j+1}$ are subgraphs of $G_{Base}(j)$, $(u, v)$ will not be in $H_j$ or $R_{j+1}$ either. 

Thus, an edge in $H_i$ or $R_{i+1}$ (it only appears in one of the two) will not be sent in any $H_j$ or $R_{j+1}$ if $j>i$. As such, $\{H_i:i\in[k]\}\cup \{R_{i+1}:i\in[k]\}$ is edge-disjoint. Since $R_1$ is empty, this set contains all the graphs sent to any player. Hence, no two players will receive the same edge.
\end{proof}

Next, we must also ensure that the maximum degree of the input graph agrees with the maximum degree given to each player. 
\begin{lemma}\label{valid-graph}
For each vertex $v\in G_{input}$, $\deg{(v)}\le \Delta$.
\end{lemma}
\begin{proof}
We consider the degree on each vertex $v$. Firstly, there is at most one value of $i$ for which $v$ is incident to some edge in $R_i$. After all, by \Cref{algo-second-step}, if an edge in $R_i$ is incident to $v$, then $v\not\in Q_{i-1}$ and $v$ is not in $G_{Base}(i)$ nor any subsequent graphs. For this value of $i$, $\Delta(R_i)\le \frac{\Delta}2$. Furthermore, $\Delta(H_j)\le \frac{2\Delta}{\distribFac}$ for all $j\in [k]$ by \Cref{compressed-partition}. Thus, $\deg(v)$ cannot exceed $\frac{\Delta}2$ in $R_i$ nor $\frac{2\Delta}{\distribFac}$ in $H_j$ for any $j\in [k]$. 

Since $G_{input}=\left(\bigcup_{j=1}^k H_j\right)\cup\left(\bigcup_{i=1}^k R_i\right)$, the degree of $v$ in $G_{input}$ is the sum of its degree in each of these graphs. In particular, the degree of $v$ is $\le k\cdot \frac{2\Delta}{\distribFac}+\frac{\Delta}2\le \frac{4\ln(n)\cdot\Delta}{\distribFac}+\frac{\Delta}2\le\frac{\Delta}2+\frac{\Delta}2=\Delta$. The second last inequality is true by $k=\lceil \ln(n)\rceil+1$ and $n>\eta$. The last is true by ${\distribFac}\ge 8\ln n$. 
\end{proof}

\Cref{graph-is-not-multi} and \Cref{valid-graph} prove that the input graph is valid. The next step is to prove that the adversary terminates --- for the sake of our analysis, we will instead show that the final base graph $G_{Base}(k)$ is small. 

\begin{lemma}\label{algorithm-finishes}
$n_k\le \frac n{\Delta^2}$. 
\end{lemma}
\begin{proof}
    We begin by proving that $n_i\le \max(\frac n{\Delta^2}, \frac{n}{e^{i-1}})$, which we will show through induction.

    For the base case, $n_i=n\le\frac n{e^0}$. 

    For the inductive step, if $n_i\le \frac n{\Delta^2}$, then $n_{i+1}\le \frac n{\Delta^2}$ too. Otherwise, by \Cref{compressed-partition}, 
    \begin{align*}
        |Q_i|\le\frac{2\ln(2)\cdot(s+1)\cdot\lfloor \frac{n_i\Delta^2}{n}\rfloor}{\frac{\Delta}{\distribFac}\cdot {\distribFac}^2\Delta}.
    \end{align*}

    Dropping the floor, the above is $\le \frac{2\ln(2)\cdot (s+1)\cdot n_i\Delta^2}{n{\distribFac}\Delta^2}=\frac{2\ln(2)\cdot(s+1)}{n{\distribFac}}\cdot n_i\le \frac 1e\cdot \frac n{e^{i-1}}=\frac n{e^i}$. 

    The second last inequality is true by the bound on ${\distribFac}$, while the last is true by the inductive hypothesis. Note that $n_{i+1}=|Q_i|$. This concludes the inductive step. \\
    
    To prove \Cref{algorithm-finishes}, we note that $\frac n{e^{k-1}}\le \frac n{e^{\ln n}}\le 1$. Since $\Delta < \sqrt n$, we have $\frac n{\Delta^2}>1>\frac n{e^{k-1}}$, so $\max\left(\frac n{\Delta^2}, \frac n{e^{k-1}}\right)=\frac n{\Delta^2}$. We showed that $n_k\le\max\left(\frac n{\Delta^2}, \frac n{e^{k-1}}\right)$, so $n_k\le \frac n{\Delta^2}$.
\end{proof}

Having proven that the input graph is valid and the algorithm terminates, it remains to bound the largest clique at each iteration. \Cref{algo-second-step} allows us to bound the largest clique in $H_{Miss}(\phi)[P]$ for all $P\in \CP_i$. The following lemma bounds the size of $\CP_i$, which will allow us to bound the size of the largest clique in the missing graph over the \textit{entirety} of $\CP_i$. 

\begin{lemma}\label{bound-on-P_i}
For all $i\in [k]$, the size of $\CP_i$ is $\le \frac{3n}{\Delta^2}$. 
\end{lemma}
\begin{proof}
By \Cref{compressed-partition}, $|\CP_i|\le\left\lceil\frac{n_i}{\lfloor\frac{n_i\Delta^2}n\rfloor}\right\rceil\le\left\lceil\frac{n_i}{\frac{n_i\Delta^2}{2n}}\right\rceil\le\left\lceil \frac{2n}{\Delta^2}\right\rceil\le\frac{2n}{\Delta^2}+1\le\frac{3n}{\Delta^2}$. 

The second inequality is true because $\frac{n_i\Delta^2}n\ge 1$. The last inequality is true by $\Delta<\sqrt n$. 
\end{proof}

Finally, we prove that Player $k$ cannot conclusively find a large independent set by proving that the adversary can find a ``breaking'' graph if the output is ever too large. 

\begin{lemma}\label{find-breaking-element}
Suppose Player $k$ outputs an Independent set $A$ of size greater than 
\begin{align*}
    \frac{n}{\Delta^2}+\frac{n}{\Delta^2}\cdot k\cdot (96{\distribFac}^2\ln(n)+30).
\end{align*}
Then, there is another graph $G_{input}'$ and set of edges to send to each player such that each player outputs the same message but an edge exists between some $u, v$ in the output of Player $k$. 
\end{lemma}
\begin{proof}
For each $v\in A$, either $v\in G_{Base}(k+1)$ or $v\in P$ for some $i\in [k]$ and $P\in \CP_i$. We let $L_i:=\bigcup_{P\in \CP_i}P$. Since $n_k<\frac n{\Delta^2}$, more than $\frac{nk(96{\distribFac}^2\ln(n)+30)}{\Delta^2}$ vertices in $A$ are also in $L_1\cup\dots\cup L_k$. By the Pigeonhole principle, for some $i\in [k]$, we have $|L_i\cap A|>\frac{n(96{\distribFac}^2\ln (n) + 30)}{\Delta^2}$. 

Since $\CP_i$ is a partition of the vertices in $L_i$, and since $|\CP_i|\le \frac{3n}{\Delta^2}$ by \Cref{bound-on-P_i}, there is some $P=(V_P, E_P)\in \CP_i$ such that $|P\cap A|>32{\distribFac}^2\ln (n)+10$ (again, by the Pigeonhole principle).

By \Cref{compressed-partition} and \Cref{algo-second-step}, the largest clique in $(H_{Miss}(M_i)-R_{i+1})[P]$ has size at most
\begin{align*}
16\ln(n)\cdot \frac{d_{filter}}{d_{removal}}+10=16\cdot \frac{{\distribFac}^2\Delta}{\Delta /2}\cdot \ln(n)+10=32{\distribFac}^2\ln(n)+10.
\end{align*}
Thus, there is some $u, v\in P$ such that $(u, v)\not\in H_{Miss}(M_i)-R_{i+1}$. 

If $(u, v)\in R_{i+1}$, then we send the same graph and edges. Note that $R_{k+1}=\emptyset$ since $n_k<\frac n{\Delta}^2$. Thus, Player $i+1$ exists and will receive $R_{i+1}$. As such, there is an edge between $(u, v)$ in $G_{input}$. 

Otherwise, we know there is some $j\in [i]$ s.t. $(u, v)\in H_{Base}(j)$ but $(u, v)\not\in H_{Miss}(M_j)$. We prove this with two cases:

If $(u, v)\in G_{Base}(i)$, then let $j=i$. Since $H_{Base}(i)[P]=G_{Base}(i)[P]$, we have $(u, v)\in H_{Base}(i)$. Additionally, since $(u, v)\not\in R_{i+1}$, we know that $(u, v)\not\in H_{Miss}(M_j)$ as well.

Otherwise, $(u, v)\not\in G_{Base}(i)$. Since $G_{Base}(i)\subseteq\dots\subseteq G_{Base}(1)$ and $(u, v)\in G_{Base}(1)$, we can find some $j<i$ s.t. $(u, v)\in G_{Base}(j)$ but $(u, v)\not\in G_{Base}(j+1)$. In particular,
\begin{align*}
(u,v)\not\in G_{Base}(j+1)=\big(G_{Base}(j)-(H_{Base}(j)-H_{Miss}(M_j))\big)[Q_j].
\end{align*}
Since $u, v\in G_{Base}(i)$ and all vertices in $G_{Base}(i)$ are also in $Q_j$, we know that $u,v\in Q_j$. This allows us to conclude that 
\begin{align*}
    (u, v)\not\in G_{Base}(j)-(H_{Base}(j)-H_{Miss}(M_j)).
\end{align*}
We know that $(u, v)\in G_{Base}(j)$, so we must also have $(u, v)\in H_{Base}(j)-H_{Miss}(M_j)$. Hence, $(u, v)\in H_{Base}(j)$ and $(u, v)\not\in H_{Miss}(M_j)$. \\

Since $(u, v)\not\in H_{Miss}(j)$, and $(u, v)\in H_{Base}(j)$, there is some graph $H'\in \mathbb{H}_i$ such that $(u, v)\in H'$ and $\Phi_i(H')=M_i$. We substitute $H_i$ with $H'$ when choosing $G_{input}$ and the edges we send each player. Each player sends the same message, but $(u, v)\in G_{input}'$, concluding the proof. 
\end{proof}

\Cref{find-breaking-element} leads naturally to a bound on the largest independent set a deterministic algorithm can find, which the following formalizes. 

\begin{lemma}\label{deterministic-lower-bound}
There does not exist a deterministic algorithm using at most $s$ bits of memory that outputs an independent set of size greater than $\frac{n}{\Delta^2}+\frac{n}{\Delta^2}\cdot k\cdot (96{\distribFac}^2\ln(n)+30)$ for all graphs with $n$ vertices and maximum degree $\Delta$.
\end{lemma}
\begin{proof}
    Suppose a deterministic algorithm exists. Then, by \Cref{algo-to-players}, there is a set of $k=\lceil \ln n\rceil +1$ players that can solve the $\IndepSet(n, \Delta, k, s)$ problem with the same independent set size. However, by \Cref{find-breaking-element}, this is impossible. After all, if Player $k$ outputs an independent set of size greater than $\frac{n}{\Delta^2}+\frac{n}{\Delta^2}\cdot k\cdot (96{\distribFac}^2\ln(n)+30)$, then we can find another set of edges to send such that each player sends the same message but the independent set is no longer valid. 
\end{proof}

To conclude, we substitute ${\distribFac}=\max\left\{\left\lceil\frac{2e\ln(2)(s+1)}n \right\rceil, \left\lceil8\ln n\right\rceil\right\}$ and $k=\lceil \ln n\rceil +1$ into the bounds we showed in \Cref{deterministic-lower-bound}.

\noindent
\textit{Proof of \Cref{main-theorem}}. This proof is simply a matter of loosening the bound in \Cref{deterministic-lower-bound} until we get a ``simpler'' form. We will start with the choices of $k$ and ${\distribFac}$. For $k$, we have
    \[
        k=\lceil\ln n\rceil+1 \le 2\ln n+1\le 3\ln n.
    \]

For the first possible value of $\distribFac$, we have
\[ 
    \left\lceil\frac{2e\ln(2)(s+1)}n \right\rceil\le \left\lceil\frac{4e\ln(2)\cdot s}n \right\rceil\le 8e\ln(2)\cdot \frac sn.
\]
For the other potential value of $\distribFac$, we have
\begin{align*}
    \lceil 8\ln n\rceil&\le 16\ln n.
\end{align*}

Hence, ${\distribFac}\le \max\{8e\ln 2\cdot \frac sn, 16\ln n\}$. Since both terms are positive and greater than $1$, we can multiply both upper bounds to get
\begin{align*}
    {\distribFac}\le 8e\ln 2\cdot \frac sn\cdot 16\ln n=128e\ln 2\cdot \frac {s\ln n}n.
\end{align*}
To simplify calculations, we will let $r=128e\ln(2)$, which is a constant.

By \Cref{deterministic-lower-bound}, no deterministic algorithm can find an independent set of size greater than
\begin{align*}
\frac{n}{\Delta^2}+\frac{n}{\Delta^2}\cdot k\cdot (96{\distribFac}^2\ln(n)+30)&\le \frac n{\Delta^2}\left(1+3\ln(n)\cdot \left[96\cdot {\left(\frac{rs\ln n}n\right)}^2\cdot \ln(n)+30\right]\right)\\
&\le \frac n{\Delta^2}\cdot 2\cdot 6{(\ln(n))}^2\left(96\cdot {\left(\frac{rs\ln n}n\right)}^2\right)\\
&=1152r^2\ln^4(n)\cdot\frac{s^2}{n\Delta^2}\\
&= \Ot\left(\frac s{\Delta^2}\cdot \frac sn\right).
\end{align*}

Hence, the size of the largest independent set that a deterministic algorithm can find for \textit{all} graphs with $n$ vertices and maximum degree $\Delta$ does not exceed $\Ot\left(\frac s{\Delta^2}\cdot \frac sn\right)$. $\qed$

\section*{Acknowledgements}
I would like to thank Sepehr Assadi for introducing me to this problem and for his numerous discussions, comments, and feedback throughout the development of this work. Without him, this would not have been possible. I would also like to thank Vihan Shah, Janani Sundaresan, and Christopher Trevisan for their insightful comments and meticulous remarks on earlier versions of this work.

\bibliography{Deterministic-Streaming-Independent-Set}

\end{document}